\documentclass[12pt,reqno,a4paper]{amsart}
\usepackage{amsmath,amssymb,dsfont,graphicx,cite,latexsym,
epsf,cancel,epic,eepic}
\usepackage{amssymb,mathrsfs}
\usepackage[colorlinks=false]{hyperref}
\usepackage{mathpazo}
\newtheorem{theorem}{Theorem}
\newtheorem{proposition}[theorem]{Proposition}

\newtheorem{lemma}[theorem]{Lemma}

\newtheorem{definition}[theorem]{Definition}
\newtheorem{cor}[theorem]{Corollary}
\def\beq{\begin{equation}}
\def\eeq{\end{equation}}
\def\bea{\begin{eqnarray}}
\def\eea{\end{eqnarray}}
\def\benpf{\noindent {\textbf{{\emph{Proof.}}\;}}}
\def\endpf{\hfill$\blacksquare$\medskip}
\makeatletter
\expandafter\let\expandafter
\reset@font\csname reset@font\endcsname
\def\subeqnarray{\arraycolsep1pt
   \def\@eqnnum\stepcounter##1{\stepcounter{subequation}
       {\reset@font\rm(\theequation\alph{subequation})}}
\jot5mm     \eqnarray}

\makeatother
\newcommand{\bbR}{{\mathbb R}}

\def\ep{\varepsilon}
\def\epsilon{\varepsilon}
\def\t{\widetilde}

\def\endpf{\hfill$\square$\medskip}

\newbox\meibox
\def\placeunder#1#2#3#4{\setbox\meibox%
\vbox{\hbox{\hskip#4$\hphantom{#2}$}\hbox{$\hphantom{#1}$}}%
\vtop{\baselineskip=0pt\lineskiplimit=\baselineskip%
\lineskip=#3\hbox to \wd\meibox{\hfil\hskip#4$#2$\hfil}%
\hbox to \wd\meibox{\hfil$#1$\hfil}}}
\def\intprod{\mathbin{\hbox to 6pt{%
                 \vrule height0.4pt width5pt depth0pt
                 \kern-.4pt
                 \vrule height6pt width0.4pt depth0pt\hss}}}

\begin{document}
%%%%%%%%%%%%%%%%%%%%%%
\title[Commuting systems of integrable symplectic birational maps]
{A construction of commuting systems \\ of integrable symplectic birational maps}
%%%%%%%%%%%%%%%%%%%%%%

\author{Matteo Petrera \and Yuri B. Suris }

\thanks{E-mail: {\tt  petrera@math.tu-berlin.de, suris@math.tu-berlin.de}}

\maketitle

\begin{center}
{\footnotesize{
Institut f\"ur Mathematik, MA 7-1\\
Technische Universit\"at Berlin, Str. des 17. Juni 136,
10623 Berlin, Germany
}}
\end{center}

%%%%%%%%%%%%%%%%%%%%%%
%%%%%%%%%%%%%%%%%%%%%%

\begin{abstract}
We give a construction of completely integrable $(2m)$-dimensional Hamiltonian systems with cubic Hamilton functions. The construction depends on a constant skew-Hamiltonian matrix $A$, that is, a matrix satisfying $A^{\rm T}J=JA$, where $J$ is a non-degenerate skew-symmetric matrix defining the standard symplectic structure on the phase space $\mathbb R^{2m}$.
Applying to any such system the so called Kahan-Hirota-Kimura discretization scheme, we arrive at a birational $(2m)$-dimensional map. We show that this map is symplectic with respect to a symplectic structure that is a perturbation of the standard symplectic structure on $\mathbb R^{2m}$, and possesses $m$ independent integrals of motion, which are perturbations of the original Hamilton functions and are in involution with respect to the invariant symplectic structure. Thus, this map is completely integrable in the Liouville-Arnold sense. Moreover, under a suitable normalization of the original $m$-tuples of commuting vector fields, their Kahan-Hirota-Kimura discretizations also commute and share the invariant symplectic structure and the $m$ integrals of motion.
\end{abstract}

%%%%%%%%%%%%%%%%%%%%%%
%%%%%%%%%%%%%%%%%%%%%%
\section{Introduction}
\label{sect intro}
%%%%%%%%%%%%%%%%%%%%%%
%%%%%%%%%%%%%%%%%%%%%%

In the recent paper \cite{PS4dim}, we introduced a large family of integrable 4-dimensional maps, along with their plenty of remarkable properties. 
We mentioned there that the construction could be probably generalized to higher dimensions. Here, such a higher dimensional generalization is achieved. We give here a construction of a big family of completely integrable Hamiltonian systems in arbitrary even dimension $2m$  and demonstrate that the so called Kahan-Hirota-Kimura discretization of these systems preserves the complete integrability and possesses a whole bunch of remarkable features. 

The set of parameters of the construction is encoded in a $2m \times 2m$ skew-Hamiltonian matrix $A$, that is, a matrix satisfying  
\begin{equation} \label{skew}
A^{\rm T}J=JA, \quad J=\begin{pmatrix} 0 & I \\ -I & 0 \end{pmatrix}.
\end{equation}
Such matrices form a $m(2m-1)$-dimensional vector space and have remarkable spectral properties \cite{skewHam}. To each non-degenerate skew-Hamiltonian matrix $A$, there corresponds a vector space of cubic polynomials $H_0(x)$ on $\mathbb R^{2m}$, satisfying a certain system of second order linear PDEs, encoded in the matrix equation 
\beq \label{harm}
A(\nabla^2 H)=(\nabla^2 H) A^{\rm T},
\end{equation}
where $\nabla^2 H$ is the Hesse matrix of the function $H$. To each such polynomial  $H_0(x)$, there corresponds a unique $(m-1)$-tuple of cubic polynomials $H_i(x)$ satisfying the same matrix differential equations \eqref{harm}, and characterized by 
\begin{equation}\label{CR}
\nabla H_i(x)=A\nabla H_{i-1}(x), \quad i=1,\ldots, m-1.
\end{equation}
For a generic $A$, the functions $H_i(x)$, $i=0,\ldots, m-1$, are functionally independent and are in involution with respect to the standard symplectic structure on $\mathbb R^{2m}$, so that the flows of the Hamiltonian vector fields $J\nabla H_i(x)$ commute. Thus, each $H_0(x)$ defines a completely integrable Hamiltonian system. 

When applied to a completely integrable Hamiltonian system $\dot{x}=J\nabla H_0(x)$ of this family, the Kahan-Hirota-Kimura discretization method produces the map $\Phi_{J\nabla H_0}$ with the following striking properties.
\begin{itemize}
\item The map $\Phi_{J\nabla H_0}$ is symplectic with respect to a symplectic structure which is a perturbation of the canonical symplectic structure on $\mathbb R^{2m}$, and possesses $m$ functionally independent integrals. In other words, $\Phi_{J\nabla H_0}$ is completely integrable. 
\item One can find integrals $\t H_i(x,\epsilon)$ $(i=0,\ldots, m-1)$ of $\Phi_{J\nabla H_0}$ which are rational perturbations of the original polynomials $H_i(x)$, are related by the same equation as \eqref{CR}, that is, $\nabla \t H_i(x,\epsilon)=A \nabla\t H_{i-1}(x,\epsilon)$, and satisfy the same second order differential equations \eqref{harm} as $H_i(x)$ do.
\item In general, the maps $\Phi_{J\nabla H_i}$ do not commute. However, one can find commuting maps as follows. We say that a linear combination of the vector fields,
$$
\sum_{i=0}^{m-1} \alpha_i J\nabla H_i=J\left(\sum_{i=0}^{m-1} \alpha_iA^i\right)\nabla H_0=JB\nabla H_0=B^{\rm T}J\nabla H_0,
$$
is {\em associated} to the vector field $J\nabla H_0$, if the skew-Hamiltonian matrix $B=\sum_{i=0}^{m-1} \alpha_iA^i$  satisfies $B^2=I$. Equivalently, the polynomial $B(\lambda)=\sum_{i=0}^{m-1} \alpha_i\lambda^i$ sends each of the $m$ distinct eigenvalues of $A$ to $\pm 1$. This defines an equivalence relation on the set of vector fields $J\nabla H$ with $H$ satisfying \eqref{harm}. It turns out that Kahan-Hirota-Kimura discretizations of associated vector fields commute and share the invariant symplectic structure and $m$ functionally independent integrals. 
\item The equivalence class of $J\nabla H_0$ consists of $2^{m-1}$ vector fields, if considered projectively (up to sign). If $m>2$, then $2^{m-1}>m$. Thus, the vector field $J\nabla H_0$ can be included (in several ways) into a system of $m$ linearly independent associated vector fields. The corresponding $m$ Kahan-Hirota-Kimura maps are independent, commute, and all share the invariant symplectic structure and $m$ functionally independent integrals.
\end{itemize}

These findings are generalizations of the corresponding results for $m=2$ given in \cite{PS4dim}. To keep notation as simple and transparent as possible, we give in the present paper a detailed account of the case $m=3$ only. The general case can be treated along the same lines without difficulties.

We give a quick review of the Kahan-Hirota-Kimura discretization method in Sect. \ref{sect HK}.
Then we discuss details of the general construction of completely integrable Hamiltonian systems generated by a skew-Hamiltonian matrix $A$ in Sect. \ref{sect int fam}. The rich algebraic properties of the corresponding vector fields are collected in Sect. \ref{sect vector fields}. Associated vector fields are introduced in Sect. \ref{sect assoc}. We prove the main results in Sect. \ref{sect commut} (commutativity), Sect. \ref{sect integrals} (integrals of motion), Sect. \ref{sect symplectic} (invariant symplectic structure) and Sect. \ref{sect diff eqs} (differential equations for the integrals of the maps).

%%%%%%%%%%%%%%%%%%%%%%
%%%%%%%%%%%%%%%%%%%%%%
\section{General properties of the Kahan-Hirota-Kimura discretization}
\label{sect HK}
%%%%%%%%%%%%%%%%%%%%%%
%%%%%%%%%%%%%%%%%%%%%%

Here we recall the main facts about the Kahan-Hirota-Kimura discretization.

This method was introduced in the geometric integration literature by Kahan in the unpublished notes \cite{K} as a method applicable to any system of ordinary differential equations for $x:\bbR\to\bbR^n$ with a quadratic vector field:
\begin{equation}\label{eq: diff eq gen}
\dot{x}=f(x)=Q(x)+Bx+c,
\end{equation}
where each component of $Q:\bbR^n\to\bbR^n$ is a quadratic form, while $B\in{\rm Mat}_{n\times n}(\bbR)$ and $c\in\bbR^n$. Kahan's discretization (with stepsize $2\epsilon$) reads as
\begin{equation}\label{eq: Kahan gen}
\frac{\widetilde{x}-x}{2\epsilon}=Q(x,\widetilde{x})+\frac{1}{2}B(x+\widetilde{x})+c,
\end{equation}
where
\[
Q(x,\widetilde{x})=\frac{1}{2}\big(Q(x+\widetilde{x})-Q(x)-Q(\widetilde{x})\big)
\]
is the symmetric bilinear form corresponding to the quadratic form $Q$. Equation (\ref{eq: Kahan gen}) is {\em linear} with respect to $\widetilde x$ and therefore defines a {\em rational} map $\widetilde{x}=\Phi_f(x,\epsilon)$. Clearly, this map approximates the time $2\epsilon$ shift along the solutions of the original differential system. Since equation (\ref{eq: Kahan gen}) remains invariant under the interchange $x\leftrightarrow\widetilde{x}$ with the simultaneous sign inversion $\epsilon\mapsto-\epsilon$, one has the {\em reversibility} property
\begin{equation}\label{eq: reversible}
\Phi_f^{-1}(x,\epsilon)=\Phi_f(x,-\epsilon).
\end{equation}
In particular, the map $f$ is {\em birational}. The explicit form of the map $\Phi_f$ defined by \eqref{eq: Kahan gen} is 
\beq \label{eq: Phi gen}
\t x =\Phi_f(x,\ep)= x + 2\ep \left( I - \ep f'(x) \right)^{-1} f(x),
\eeq
where $f'(x)$ denotes the Jacobi matrix of $f(x)$. Moreover, if the vector field $f(x)$ is homogeneous (of degree 2), then \eqref{eq: Phi gen} can be equivalently rewritten as
\beq \label{eq: Phi hom}
\t x =\Phi_f(x,\ep)= \left( I - \epsilon f'(x) \right)^{-1} x.
\eeq
Due to \eqref{eq: reversible}, in the latter case we also have:
\beq \label{eq: Phi hom alt}
x =\Phi_f(\t x,-\ep)= \left( I + \epsilon f'(\t x) \right)^{-1} \t x \quad \Leftrightarrow \quad \t x = \left( I + \epsilon f'(\t x) \right) x.
\eeq
One has the following expression for the Jacobi matrix of the map $\Phi_f$:
\beq \label{Jac}
d\Phi_f(x)=\frac{\partial\t x}{\partial x}=\big(I-\epsilon f'(x)\big)^{-1}\big(I+\epsilon f'(\t x)\big).
\eeq

Kahan applied this discretization scheme to the famous Lotka-Volterra system and showed that in this case it possesses a very remarkable non-spiralling property. This property was explained by Sanz-Serna \cite{SS} by demonstrating that in this case the numerical method preserves an invariant Poisson structure of the original system.

The next intriguing appearance of this discretization was in two papers by Hirota and Kimura who (being apparently unaware of the work by Kahan) applied it to two famous {\em integrable} system of classical mechanics, the Euler top and the Lagrange top \cite{HK, KH}. Surprisingly, the discretization scheme produced in both cases {\em integrable} maps. 

In \cite{PS, PPS1, PPS2} the authors undertook an extensive study of the properties of the Kahan's method when applied to integrable systems (we proposed to use in the integrable context the term ``Hirota-Kimura method''). It was demonstrated that, in an amazing number of cases, the method preserves integrability in the sense that the map $\Phi_f(x,\epsilon)$ possesses as many independent integrals of motion as the original system $\dot x=f(x)$.

Further remarkable geometric properties of the Kahan's method were discovered by Celledoni, McLachlan, Owren and Quispel in \cite{CMOQ1}, see also \cite{CMOQ2, CMOQ3}. These properties are unrelated to integrability. They demonstrated that for an arbitrary Hamiltonian vector field $f(x)=J\nabla H(x)$ with a constant Poisson tensor $J$ and a cubic Hamilton function $H(x)$, the map $\Phi_f(x,\epsilon)$ possesses a rational integral of motion 
\beq\label{eq: CMOQ integral}
\t H(x,\epsilon) = \frac{1}{6\epsilon}x^{\rm T}J^{-1}\t x=H(x)+\frac{2\ep}{3}  (\nabla H(x))^{\rm T} \left( I - \ep f'(x) \right)^{-1} f(x),
\eeq
as well as an invariant measure
\beq\label{eq: CMOQ measure}
\frac{dx_1\wedge\ldots\wedge dx_n}{\displaystyle{\det\left( I - \ep f'(x) \right)}}.
\eeq
It should be mentioned that, while for $n=2$ the existence of an invariant measure is equivalent to symplecticity of the map $\Phi_f(x,\epsilon)$, the latter property was not established (before our paper \cite{PS4dim}) for any quadratic Hamiltonian system  in dimension $n\ge 4$.

%%%%%%%%%%%%%%%%%%%%%%
%%%%%%%%%%%%%%%%%%%%%%
\section{A family of integrable 6-dimensional Hamiltonian systems}
\label{sect int fam}
%%%%%%%%%%%%%%%%%%%%%%
%%%%%%%%%%%%%%%%%%%%%%

% In this Section we propose a simple way to construct a large family of integrable 4-dimensional Hamiltonian systems.

Consider the canonical phase space $\bbR^{6}$ with coordinates $(x_1,\ldots,x_{6})$, equipped with the standard symplectic structure in which the Poisson brackets of the coordinate functions are $\{x_1,x_3\}= \{x_2,x_4\}= \{x_3,x_6\}=1$ (all other brackets being either obtained from these ones by skew-symmetry or otherwise vanish). Let $H_0(x)$  be a Hamilton function on  $\bbR^{6}$. The corresponding Hamiltonian system is governed by the equations of motion
\beq  \label{Ham}
\dot x= J \nabla H_0(x),\quad J=\begin{pmatrix} 0 & I \\ -I & 0 \end{pmatrix}.
\eeq
\begin{proposition}
Consider a constant non-degenerate $6\times 6$ matrix $A$, and suppose that for functions $H_1(x)$, $H_2(x)$ the following relations are satisfied:
\begin{eqnarray}
\nabla H_1 & = & A \nabla H_0,  \label{grad K}\\
\nabla H_2 & = & A \nabla H_1.  \label{grad L}
\end{eqnarray}
If the matrix $A$ is skew-Hamiltonian, that is, if it satisfies
\beq \label{cond A}
 A^{\rm{T}} J = J A,
\eeq
then the functions $H_0$, $H_1$, $H_2$ are pairwise in involution, so that the Hamiltonian system \eqref{Ham} is completely integrable.
\end{proposition}

\benpf We have:
\begin{eqnarray*}
\{H_0,H_1\} & = & (\nabla H_0)^{\rm{T}} J \nabla H_1 = (\nabla H_0)^{\rm{T}} J A \nabla H_0, \\
\{H_1,H_2\} & = & (\nabla H_1)^{\rm{T}} J \nabla H_2=(\nabla H_1)^{\rm{T}} J A \nabla H_1,
\end{eqnarray*}
and both expressions vanish if the matrix $JA$ is skew-symmetric, which gives condition \eqref{cond A}. Since $\nabla H_2=A^2\nabla H_0$, and the matrix $A^2$ satisfies the same condition \eqref{cond A}:
$$
(A^{\rm T})^2J=A^{\rm T}JA=JA^2,
$$
we see that also $\{H_0,H_2\}=0$. Equations \eqref{grad K}, \eqref{grad L} with a non-scalar matrix $A$ also ensure that $\nabla H_0$, $\nabla H_1$, $\nabla H_2$ are linearly independent (generically).
\endpf

If $A$ is written in the block form as
$$
A = \begin{pmatrix}
A_1 & A_2 \\
A_3 & A_4
\end{pmatrix}
$$
with $3\times 3$ blocks $A_i$, then the skew-Hamiltonian condition \eqref{cond A} reads:
\beq \label{cond2}
A_1^{\rm{T}}=A_4, \qquad A_2^{\rm{T}}=-A_2, \qquad A_3^{\rm{T}}=-A_3.
\eeq
Such matrices form a 15-dimensional vector space:
\beq \label{A elem}
A=\begin{pmatrix}
a_1 & a_2 & a_3 & 0 & -a_{10} & -a_{11}  \\
a_4 & a_5 & a_6 & a_{10} & 0 & -a_{12}  \\
a_7 & a_8 & a_9 & a_{11} & a_{12} & 0  \\
0 & a_{13} & a_{14} & a_1 & a_4 & a_7 \\
-a_{13} & 0 & a_{15} & a_2 & a_5 & a_8 \\
-a_{14} & -a_{15} & 0 & a_3 & a_6 & a_9
\end{pmatrix}.
\eeq

We now discuss applicability of this construction. For a given function $H_0$, differential equations \eqref{grad K} for $H_1$ are solvable if and only if $H_0$ satisfies the following condition:
\beq  \label{d2H}
A(\nabla^2 H)=(\nabla^2 H) A^{\rm T},
\eeq
where $\nabla^2 H$ is the Hesse matrix of the function $H$. Note that solutions $H_1(x)$ of \eqref{grad K} satisfy the same compatibility conditions \eqref{d2H}. To see this, observe that one has $\nabla H_0=A^{-1} \nabla H_1$, and the solvability conditions for this equation in terms of $H_1$ is the same as \eqref{d2H} but with $A$ replaced by $A^{-1}$:
$$
A^{-1}(\nabla^2 H_1)=(\nabla^2 H_1) (A^{-1})^{\rm T},
$$
which is clearly equivalent to 
$$
A(\nabla^2 H_1)=(\nabla^2 H_1) A^{\rm T}.
$$
The latter equation ensures solvability of \eqref{grad L} for $H_2$, and also the fact that solutions $H_2$ satisfy the same system of second order differential equations \eqref{d2H}. Thus, to any solution $H_0$ of the system  \eqref{d2H} there correspond, via \eqref{grad K}, \eqref{grad L}, two other solutions $H_1$, $H_2$ (unique up to additive constants). 

\begin{proposition}
The linear space of homogeneous polynomials on $x_1,\ldots,x_6$ of degree 3, satisfying the system of second order PDEs \eqref{d2H}, has dimension 12.
\end{proposition}
\begin{proof}
In general, condition $A (\nabla^2 H) = (\nabla^2 H ) A^T$ gives 15 scalar PDEs $C_i=0$. 
Only 12 of them are linearly independent.
A general homogeneous polynomial of degree 3 in six variables $x_1, \ldots,x_6$ has 56 coefficients. 
Each of the expressions $C_i$ , $i = 1, \ldots, 12$, is a linear polynomial in these six variables, so that each equation $C_i = 0$ results in 6 linear equations for the coefficients of $H$. Altogether we get $6\cdot12=72$ linear homogeneous equations for 56 coefficients of $H$. A careful inspection of the resulting linear system reveals that the rank of its matrix is equal to 44. Therefore, the dimension of the space of solutions is equal to $56 - 44 = 12$. 

\end{proof}

%%%%%%%%%%%%%%%%%%%%%%%%%
%%%%%%%%%%%%%%%%%%%%%%%%%
\section{General algebraic properties of the vector fields $f_0$, $f_1$, $f_2$}
\label{sect vector fields}
%%%%%%%%%%%%%%%%%%%%%%%%%
%%%%%%%%%%%%%%%%%%%%%%%%%

From now on we will assume that $A$ is a skew-Hamiltonian matrix, that is,  satisfies \eqref{cond A}. 
Assume that $H_0(x)$, $H_1(x)$ and $H_2(x)$ are homogeneous polynomials of degree 3 satisfying \eqref{grad K}, \eqref{grad L}. Set 
\begin{eqnarray}
f_0(x) & = & J\nabla H_0(x),  \label{f}\\
f_1(x) & =& J\nabla H_1(x)=JA\nabla H_0(x),  \label{g} \\
f_2(x) & =& J\nabla H_2(x)=JA^2\nabla H_0(x).  \label{h}
\end{eqnarray}
Due to \eqref{cond A}, we have:
\begin{eqnarray} 
f_1(x) & = & A^{\rm T} f_0(x), \label{f vs g} \\
f_2(x) & = & A^{\rm T} f_1(x) = (A^{\rm T})^2 f_0(x). \label{f vs h} 
\end{eqnarray}
Furthermore, we have:
\begin{eqnarray} 
f_1'(x) & = & A^{\rm T} f_0'(x), \label{f' vs g'} \\
% \quad \Leftrightarrow \quad f(x)=\alpha^{-1}A^{\rm T}g(x).
f_2'(x) & = & A^{\rm T} f_1'(x) = (A^{\rm T})^2 f_0'(x). \label{f' vs h'} 
\end{eqnarray}
\begin{lemma}
The following identities hold true:
\beq \label{Ham f'}
(f_i'(x))^{\rm T}J=-Jf_i'(x) \quad (i=0,1,2),
\eeq
\beq \label{Af'}
 A^{\rm T}f_i'(x)=f_i'(x)A^{\rm T} \quad (i=0,1,2),
\eeq
\beq \label{vf comm}
 f_i'(x)f_j(x)=f_j'(x)f_i(x) \quad (i,j=0,1,2),
\eeq
\beq \label{f'g'=g'f'}
 f_i'(x)f_j'(x)=f_j'(x)f_i'(x) \quad (i,j=0,1,2).
\eeq
\end{lemma}
\begin{proof}
Equation \eqref{Ham f'} is the characteristic property of Jacobi matrices of Hamiltonian vector fields.
Equation \eqref{Af'} is equivalent to \eqref{d2H}, due to \eqref{cond A}:
\[
 A^{\rm T}f_i'(x)=A^{\rm T}J\nabla^2 H_i(x)=JA\nabla^2 H_i(x)=J\nabla^2 H_i(x) A^{\rm T}=f_i'(x)A^{\rm T}.
\]

To prove \eqref{vf comm}, we compute with the help of \eqref{f vs g}, \eqref{Af'}, \eqref{f' vs g'}:
\[
f_0'(x)f_1(x)=f_0'(x)A^{\rm T}f_0(x)=A^{\rm T}f_0'(x)f_0(x)=f'_1(x)f_0(x),
\]
and similarly for  other two relations. Observe that \eqref{vf comm} expresses the pairwise commutativity of the vector fields $f_i(x)$. 

Identities \eqref{f'g'=g'f'} are proved along the same lines, with the help of \eqref{f' vs g'}, \eqref{Af'}:
\[
f_0'(x)f'_1(x)=f'_0(x)A^{\rm T}f'_0(x)=A^{\rm T}f'_0(x)f'_0(x)=f_1'(x)f'_0(x),
\]
and similarly for  other two relations.
\end{proof}

%%%%%%%%%%%%%%%%%%%%%%%%%%%%%%%%%%%%
\section{Associated vector fields}
\label{sect assoc}
%%%%%%%%%%%%%%%%%%%%%%%%%%%%%%%%%%%%%%

\begin{definition} \label{def assoc}
Let the skew-Hamiltonian matrix 
$$
B=\alpha I+\beta A+\gamma A^2
$$
satisfy 
$$
B^2=I.
$$
Then the vector field
$$
g(x) =JB \nabla H_0(x)=B^{\rm T}J\nabla H_0(x)=B^{\rm T} f_0(x)
$$
is called {\em associated} to the vector field $f_0(x)$. The vector field $g(x)$ is Hamiltonian,
$$
g(x) = J\nabla K(x),
$$
with the Hamilton function
$$
K(x) =  \alpha H_0(x)+\beta H_1(x)+\gamma H_2(x).
$$
\end{definition}
This defines an equivalence relation on the set of vector fields $J\nabla H(x)$ with the Hamilton functions $H(x)$ satisfying \eqref{d2H}. 

\begin{lemma} \label{lemma g'g'}
If vector field $g(x)$ is associated to $f_0(x)$ via the matrix $B$, then the following identities hold true:
\beq \label{f'f=g'g}
 g'(x)g(x)=f_0'(x)f_0(x),
\eeq
\beq \label{f'f'=g'g'}
(g'(x))^2=(f_0'(x))^2.
\eeq
\end{lemma}
\begin{proof}
We first check \eqref{f'f=g'g}:
\[
g'(x)g(x)=g'(x)B^{\rm T}f_0(x)=B^{\rm T}g'(x)f_0(x)=(B^{\rm T})^2 f_0'(x)f_0(x)= f_0'(x)f_0(x).
\]
For \eqref{f'f'=g'g'} everything is similar:
\[
g'(x)g'(x)=g'(x)B^{\rm T}f_0'(x)=B^{\rm T}g'(x)f_0'(x)=(B^{\rm T})^2 (f_0'(x))^2= (f_0'(x))^2.
\]
\end{proof}

\begin{cor}\label{th1}
Under conditions of Lemma \ref{lemma g'g'}, we have:
\beq \label{condet}
 \det  \left(I - \ep g'(x) \right)=\det  \left(I - \ep f_0'(x) \right).
\eeq
\end{cor}
\begin{proof}
It is enough to prove that ${\rm tr}\, (f_0'(x))^k={\rm tr}\, (g'(x))^k$ for all $k$. Since vector fields $f_0$, $g$ are Hamiltonian, we have:
${\rm tr}\, (f_0'(x))^k=0$ and ${\rm tr}\, (g'(x))^k=0$ for odd $k$. The equalities for even $k$ follow from \eqref{f'f'=g'g'}.
\end{proof}

\begin{lemma}
For a generic $6\times 6$ matrix skew-Hamiltonian $A$, there exist three linearly independent skew-Hamiltonian matrices
\begin{equation}
B_i=\alpha_i I+\beta_i A+\gamma_i A^2\quad (i=1,2,3),
\end{equation}
satisfying
\begin{equation}\label{B^2}
B_i^2=I,
\end{equation}
and such that 
\begin{equation}\label{B char pol}
\det(B_i-\lambda I)=(\lambda-1)^4(\lambda+1)^2.
\end{equation}
They are related by
\begin{equation}\label{B sum}
B_1+B_2+B_3=I,
\end{equation}
which is equivalent to
\begin{equation}
\alpha_1+\alpha_2+\alpha_3=1, \quad \beta_1+\beta_2+\beta_3=0, \quad \gamma_1+\gamma_2+\gamma_3=0.
\end{equation}
\end{lemma}
\begin{proof} It was demonstrated in \cite{skewHam} that any skew-Hamiltonian matrix $A$ is related by a symplectic similarity transformation to a matrix of the form
$$
SAS^{-1}=    \begin{pmatrix} W & 0 \\ 0 & W^{\rm T} \end{pmatrix}, \quad S^{\rm T}JS=J.
$$
As a consequence, the characteristic polynomial of a skew-Hamiltonian $6\times 6$ matrix $A$ is of the form 
$$
\det(A-\lambda I)=(\lambda^3-a\lambda^2+b\lambda-c)^2.
$$
Thus, a generic skew-Hamiltonian $6\times 6$ matrix $A$ has three double eigenvalues $\lambda_i$, $i=1,2,3$, which are pairwise distinct (the latter being the definition of ``generic''). We have:
\begin{equation}\label{abc}
a=\lambda_1+\lambda_2+\lambda_3, \quad b=\lambda_1\lambda_2+\lambda_2\lambda_3+\lambda_3\lambda_1, \quad
c=\lambda_1\lambda_2\lambda_3.
\end{equation}
Moreover, each $\lambda_i$ has two linearly independent eigenvectors, so that $A$ is diagonalizalbe.

We construct (according to the Lagrange interpolating formula) three quadratic polynomials
$$
B_i(\lambda)=\alpha_i+\beta_i\lambda+\gamma_i\lambda^2 \quad (i=1,2,3)
$$
such that
\begin{equation}\label{B constr}
B_i(\lambda_i)=-1, \quad B_i(\lambda_j)=1, \quad B_i(\lambda_k)=1, \quad {\rm where} \quad \{j,k\}=\{1,2,3\}\setminus\{i\}.
\end{equation}
Thus, each matrix $B_i=B_i(A)$ has two double eigenvalues equal to 1 and one double eigenvalue equal to $-1$, so that \eqref{B char pol} is satisfied. As a corollary, each matrix $B_i^2$ has all six eigenvalues equal to 1, which, together with diagonalizabilty, yields \eqref{B^2}. Finally, from \eqref{B constr} there follows that the quadratic polynomial $B_1(\lambda)+B_2(\lambda)+B_3(\lambda)$ takes the value 1 at the three points $\lambda_1$, $\lambda_2$, $\lambda_3$, therefore it is identically equal to 1. This yields \eqref{B sum}.
\end{proof}

For further reference, we give the formulas for the coefficients of the Lagrange interpolation polynomials $B_i(\lambda)$ through the eigenvalues of $A$:
\begin{eqnarray}
\alpha_i & = & \frac{\lambda_i^2-b}{(\lambda_i-\lambda_j)(\lambda_i-\lambda_k)}, \label{alpha} \\
\beta_i & = &  \frac{2(\lambda_j+\lambda_k)}{(\lambda_i-\lambda_j)(\lambda_i-\lambda_k)}, \label{beta} \\
\gamma_i & = & -\frac{2}{(\lambda_i-\lambda_j)(\lambda_i-\lambda_k)}. \label{gamma}
\end{eqnarray}

Thus, the equivalence class of $f_0(x)=J\nabla H_0(x)$ consists of four vector fields $f_0(x)$ and $g_i(x)$, $i=1,2,3$, if considered projectively (up to sign). We mention the following relations for the matrices $B_i$ which are direct consequences of defining relations \eqref{B constr}:
\begin{equation}
B_1B_2=B_2B_1=-B_3,\quad B_2B_3=B_3B_2=-B_1,\quad B_3B_1=B_1B_3=-B_2.
\end{equation}
Due to this formula, the triple of vector fields associated to $g_1(x)$ (say) is
\begin{equation}\label{assoc triple}
B_1^{\rm T} g_1(x)=f_0(x), \quad B_2^{\rm T} g_1(x)=-g_3(x), \quad B_3^{\rm T} g_1(x)=-g_2(x).
\end{equation}
Our main results are the following: for two associated vector fields $f$ and $g$ the Kahan maps $\Phi_f$ and $\Phi_g$ commute (Theorem \ref{th commute}), share three independent integrals of motion (Theorem \ref{th integrals}), and share an invariant symplectic structure (Theorem \ref{th Poisson}).

\begin{lemma} \label{lemma f'f'}
The following identity holds true:
\begin{equation}  \label{f'f' thru B}
(f'_0(x))^2  =  p_0(x)I+\sum_{i=1}^3 q_i(x)B_i^{\rm T} ,
\end{equation}
where
\begin{eqnarray}
p_0(x)  & = & \frac{1}{8}{\rm tr} (f'_0(x))^2,  \label{p0}\\
q_i(x) & = & \frac{1}{8}{\rm tr}\big(B_i^{\rm T}(f'_0(x))^2\big)=\frac{1}{8}{\rm tr}\big(f'_0(x)g'_i(x)\big). \label{qi}
\end{eqnarray}
\end{lemma}
\begin{proof}
Due to $f'_0-g'_i=(I-B_i^{\rm T})f'_0$ and to the fact that $B_i$ has a quadruple eigenvalue 1, we have: 
\begin{equation}\label{char pol f'-g'}
\det(\lambda I- (f'_0-g'_i)) = \lambda^6-\frac{1}{2}\lambda^4 \ {\rm tr} (f'_0-g'_i)^2 \quad (i=1,2,3). 
\end{equation}
By the theorem of Caley-Hamilton, we have:
\begin{equation} \label{pol f'-g'}
(f'_0-g'_i)^6-\frac{1}{2}(f'_0-g'_i)^4 \ {\rm tr} (f'_0-g'_i)^2 =0 \quad (i=1,2,3).
\end{equation}
Taking into account \eqref{f'g'=g'f'} and \eqref{f'f'=g'g'}, these equations result in
$$
32 (f'_0)^6 -32B_i^{\rm T}(f'_0)^6-\big(8 (f'_0)^4 -8B_i^{\rm T}(f'_0)^4\big)\big( {\rm tr} (f'_0)^2- {\rm tr}(B_i^{\rm T}(f'_0)^2)\big)=0 \quad (i=1,2,3).
$$
Upon dividing by a generically non-degenerate matrix $(f'_0)^4$, we arrive at 
$$
32 (f'_0)^2 -32B_i^{\rm T}(f'_0)^2-8(I -B_i^{\rm T})\big( {\rm tr} (f'_0)^2- {\rm tr}(B_i^{\rm T}(f'_0)^2)\big)=0 \quad (i=1,2,3).
$$
Upon summation over $i=1,2,3$ and taking into account that $B_1+B_2+B_3=I$, we find:
$$
64(f'_0)^2-8{\rm tr}((f'_0)^2)\ I-8\sum_{i=1}^3{\rm tr}(B_i^{\rm T}(f'_0)^2) \ B_i^{\rm T}=0. 
$$
This is equivalent to \eqref{f'f' thru B}. 
\end{proof}

%%%%%%%%%%%%%%%%%%%%%
%%%%%%%%%%%%%%%%%%%%%
\section{Commutativity of maps}
\label{sect commut}
%%%%%%%%%%%%%%%%%%%%%%
%%%%%%%%%%%%%%%%%%%%%%%
In this section, $f=f_0$ and $g$ are two associated vector fields, via the skew-Hamiltonian matrix $B$.

\begin{theorem}\label{th commute}
The maps
\bea
&& \Phi_{f}: x \mapsto \t x = 
\left( I - \ep f'(x) \right)^{-1} x =
\left( I + \ep f'(\t x) \right) x, \label{eq: Phi1} \\
&& \Phi_{g}: x \mapsto \widehat x = 
\left( I - \ep g'(x) \right)^{-1} x =
\left( I + \ep g'(\widehat x) \right) x, \label{eq: Phi2} 
\eea
commute:  $\Phi_{f} \circ \Phi_{g}=\Phi_{g} \circ \Phi_{f}$.
\end{theorem}

\begin{proof}
We have:
\begin{equation}\label{eq: hat tilde}
\left(\Phi_{g} \circ \Phi_{f}\right)(x)=\left( I - \ep g'(\t x) \right)^{-1} \left( I + \ep f'(\t x) \right) x,
\end{equation}
and
\begin{equation}\label{eq: tilde hat}
\left(\Phi_{f} \circ \Phi_{g}\right)(x)=\left( I - \ep f'(\widehat x) \right)^{-1} \left( I + \ep g'(\widehat x) \right) x.
\end{equation}
We prove the following matrix equation:
\beq \label{det}
\left( I - \ep g'(\t x) \right)^{-1} \left( I + \ep f'(\t x) \right) =\left( I - \ep f'(\widehat x) \right)^{-1} \left( I + \ep g'(\widehat x) \right),
\eeq
which is stronger than the vector equation $\left(\Phi_{f}\circ\Phi_{g}\right)(x)=\left(\Phi_{g}\circ\Phi_{f}\right)(x)$ expressing commutativity. Equation \eqref{det} is equivalent to
\beq \label{det 1}
\left( I - \ep f'(\widehat x) \right) \left( I - \ep g'(\t x) \right)^{-1} =
  \left( I + \ep g'(\widehat x) \right)\left( I + \ep f'(\t x) \right)^{-1}.
\eeq
From \eqref{f'f'=g'g'} we find:
\[
\left( I - \ep g'(\t x) \right)^{-1}=\left( I + \ep g'(\t x) \right)\left( I - \ep^2 (f'(\t x))^2 \right)^{-1},
\]
\[
\left( I + \ep f'(\t x) \right)^{-1}=\left( I - \ep f'(\t x) \right)\left( I - \ep^2(f'(\t x))^2 \right)^{-1}.
\]
With this at hand, equation \eqref{det 1} is equivalent to
\[
\left( I - \ep f'(\widehat x) \right)\left( I + \ep g'(\t x) \right)=
\left( I + \ep g'(\widehat x) \right)\left( I - \ep f'(\t x) \right).
\]
Here the quadratic in $\epsilon$ terms cancel by virtue of \eqref{f' vs g'} and \eqref{Af'}:
\[
f'(\widehat x)g'(\t x) = f'(\widehat x)B^{\rm T}f'(\t x) =  B^{\rm T}f'(\widehat x)f'(\t x)  =  g'(\widehat x)f'(\t x),
\]
so that we are left with the terms linear in $\epsilon$:
\beq \label{aux}
-  f'(\widehat x)+ g'(\t x)= g'(\widehat x)- f'(\t x).
\eeq
Since the tensors $f''$, $g''$ are constant, we have:
\[
f'(\widehat x)=f'(x)+f''(\widehat x-x)=f'(x)+2\ep f''\left(I-\ep g'(x)\right)^{-1}g(x),
\]
\[
g'(\widehat x)=g'(x)+g''(\widehat x-x)=g'(x)+2\ep g''\left(I-\ep g'(x)\right)^{-1}g(x),
\]
\[
f'(\t x)=f'(x)+f''(\t x-x)=f'(x)+2\ep f''\left(I-\ep f'(x)\right)^{-1}f(x),
\]
\[
g'(\t x)=g'(x)+g''(\t x-x)=g'(x)+2\ep g''\left(I-\ep f'(x)\right)^{-1}f(x).
\]
Thus, equation \eqref{aux} is equivalent to 
\begin{align} \label{aux2}
& f''\left(I-\ep g'(x)\right)^{-1}g(x)+g''\left(I-\ep g'(x)\right)^{-1}g(x) = \nonumber\\
& \qquad f''\left(I-\ep f'(x)\right)^{-1}f(x)+g''\left(I-\ep f'(x)\right)^{-1}f(x).
\end{align}
At this point, we use the following statement.

\begin{lemma} \label{lemma g''}
For any vector $v\in\mathbb C^6$ we have:
\beq \label{g''}
 g''(x)v=f''(x)(B^{\rm T}v), \quad  f''(x)v=g''(x)(B^{\rm T}v).
\eeq
\end{lemma}

\noindent
We compute the matrices on the left-hand side of \eqref{aux2} with the help of \eqref{g''}, \eqref{f vs g}, \eqref{f' vs g'}:
\begin{eqnarray*}
f''\left(I-\ep g'(x)\right)^{-1}g(x) & = & f''\left(I-\ep^2(f'(x))^2\right)^{-1}\left(g(x)+\ep g'(x)g(x)\right),\\
g''\left(I-\ep g'(x)\right)^{-1}g(x) & = & f''\left(I-\ep^2(f'(x))^2\right)^{-1}B^{\rm T}\left(g(x)+\ep g'(x)g(x)\right)\\
                                                              & = & f''\left(I-\ep^2 (f'(x))^2\right)^{-1}\left(f(x)+\ep f'(x)g(x)\right),
\end{eqnarray*}
and similarly    
\begin{eqnarray*}
f''\left(I-\ep f'(x)\right)^{-1}f(x) & = & f''\left(I-\ep^2 (f'(x))^2\right)^{-1}\left(f(x)+\ep f'(x)f(x)\right)\\
g''\left(I-\ep f'(x)\right)^{-1}f(x) & = &f''\left(I-\ep^2 (f'(x))^2\right)^{-1}B^{\rm T}\left(f(x)+\ep f'(x)f(x)\right)\\
                                                              & = & f''\left(I-\ep^2 (f'(x))^2\right)^{-1}\left(g(x)+\ep g'(x)f(x)\right).
\end{eqnarray*}                                                          
Collecting all the results and using \eqref{vf comm} and \eqref{f'f=g'g}, we see that the proof is complete. 
\end{proof}

{\em Proof of Lemma \ref{lemma g''}.} The identities in question are equivalent to
\beq \label{Af''}
 B^{\rm T}(f''(x)v)=f''(x)(B^{\rm T}v), \quad  B^{\rm T}(g''(x)v)=g''(x)(B^{\rm T}v).
\eeq
(Actually, both tensors $f''$ and $g''$ are constant, i.e., do not depend on $x$.) To prove the latter identities, we start with equation \eqref{Af'} written in components:
\[
\sum_k (B^{\rm T})_{ik}\frac{\partial f_k}{\partial x_\ell}=\sum_k \frac{\partial f_i}{\partial x_k}(B^{\rm T})_{k\ell}.
\]
Differentiating with respect to $x_j$, we get:
\[
\sum_k (B^{\rm T})_{ik}\frac{\partial^2 f_k}{\partial x_j\partial x_\ell}=\sum_k \frac{\partial f_i}{\partial x_j\partial x_k}(B^{\rm T})_{k\ell}.
\]
Hence,
\[
\sum_{k,\ell} (B^{\rm T})_{ik}\frac{\partial^2 f_k}{\partial x_j\partial x_\ell}v_\ell=\sum_{k,\ell} \frac{\partial f_i}{\partial x_j\partial x_k}(B^{\rm T})_{k\ell}v_\ell,
\]
which is nothing but the $(i,j)$ entry of the matrix identity \eqref{Af''}. \qed

%%%%%%%%%%%%%%%%%%%%%%
%%%%%%%%%%%%%%%%%%%%%%
\section{Integrals of motion}
\label{sect integrals}
%%%%%%%%%%%%%%%%%%%%%%%
%%%%%%%%%%%%%%%%%%%%%%%
Also in this section, $f=f_0$ and $g$ are two associated vector fields, via the skew-Hamiltonian matrix $B$.

\begin{theorem} \label{th integrals}
The maps $\Phi_{f}$ and $\Phi_{g}$ share two functionally independent conserved quantities 
\begin{equation}\label{tH}
\t H(x,\epsilon)=\epsilon^{-1}x^{\rm T}J\ \t x=\epsilon^{-1}x^{\rm T}J\left(I-\ep f'(x)\right)^{-1}x
\end{equation}
and 
\begin{equation}\label{tK}
\t K(x,\epsilon)=\epsilon^{-1} x^{\rm T}J\ \widehat x=\epsilon^{-1} x^{\rm T}J\left(I-\ep g'(x)\right)^{-1}x.
\end{equation}
\end{theorem}
Before proving this theorem, we observe different expressions for these functions. Expanding \eqref{tH} in power series with respect to $\epsilon$, we see:
$$
\t H(x, \epsilon) = \sum_{k=0}^\infty \epsilon^{k-1} x^{\rm T}J(f'(x))^kx=-\sum_{k=0}^\infty \epsilon^{k-1}x^{\rm T} \Big(\nabla^2 H(x) J\nabla^2 H(x) \cdots J\nabla^2 H(x)\Big) x.
$$
The matrix in the parentheses (involving $k$ times $\nabla^2 H(x)$ and $k-1$ times $J$) is symmetric if $k$ is odd, and skew-symmetric if $k$ is even. Therefore, all terms with even $k$ vanish, and we arrive at
$$
\t H(x, \epsilon) = \sum_{k=0}^\infty \epsilon^{2k} x^{\rm T}J(f'(x))^{2k+1}x,
$$
or, finally,
\begin{equation} \label{tH alt}
\t H(x, \epsilon) =x^{\rm T} J\left(I-\epsilon^2(f'(x))^2\right)^{-1} f'(x) x =
2x^{\rm T} J\left(I-\epsilon^2(f'(x))^2\right)^{-1} f(x) .
\end{equation}
At the last step we used that $f(x)$ is homogeneous of degree 2, so that $f'(x)x=2f(x)$.  Of course, analogous expressions hold true for the function $\t K(x,\epsilon)$:
\begin{equation} \label{tK alt}
\t K(x, \epsilon) = x^{\rm T} J\left(I-\epsilon^2(g'(x))^2\right)^{-1} g'(x) x =
2 x^{\rm T} J\left(I-\epsilon^2(g'(x))^2\right)^{-1} g(x) .
\end{equation}
Formulas \eqref{tH alt}, \eqref{tK alt} also clearly display the asymptotics
$$
\t H(x, \epsilon)=2x^{\rm T} Jf(x) +O(\epsilon^2)=-2x^{\rm T} \nabla H(x)+O(\epsilon^2)=-6H(x)+O(\epsilon^2),
$$
and analogously 
$$
\t K(x, \epsilon)=2 x^{\rm T} Jg(x) +O(\epsilon^2)=-2x^{\rm T} \nabla K(x)+O(\epsilon^2)=-6K(x)+O(\epsilon^2).
$$

{\em Proof of Theorem \ref{th integrals}.}
We first show that $\t H(x,\epsilon)$ is an integral of motion of the map $\Phi_f$ (this is a result from \cite{CMOQ1}, which holds true for arbitrary Hamiltonian vector fields). For this goal, we compute with the help of \eqref{eq: Phi1}:
\begin{eqnarray*}
\t H(\t x,\epsilon) & = & \epsilon^{-1}\t x^{\rm T}J\left(I-\ep f'(\t x)\right)^{-1}\t x  \\
 & = & \epsilon^{-1}x^{\rm T}\left(I+\epsilon f'(\t x)\right)^{\rm T}J\left(I-\ep f'(\t x)\right)^{-1}\left(I-\epsilon f'(x)\right)^{-1} x .
\end{eqnarray*}
Taking into account \eqref{Ham f'} in the form $(f'(\t x))^{\rm T} J=-Jf'(\t x)$, we arrive at
\begin{equation}
\t H(\t x,\epsilon)\; =\; \epsilon^{-1}x^{\rm T}J(I-\epsilon f'(\t x))\left(I-\ep f'(\t x)\right)^{-1}\left(I-\epsilon f'(x)\right)^{-1} x \; = \; \t H(x,\epsilon).
\end{equation}

Next, we show that $\t K(x,\epsilon)$ also is an integral of motion of the map $\Phi_f$. For this goal, we first compute, based on \eqref{tK alt}:
\begin{eqnarray*}
\t K(\t x,\epsilon) & = &  \t x^{\rm T}J\left(I-\ep^2 (g'(\t x))^2\right)^{-1}g'(\t x)\t x  \\
 & = &   x^{\rm T}\left(I+\epsilon f'(\t x)\right)^{\rm T}J\left(I-\ep^2 (g'(\t x))^2\right)^{-1}g'(\t x)\left(I+\epsilon f'(\t x)\right) x \\
  & = &   x^{\rm T}J\left(I-\epsilon f'(\t x)\right)\left(I-\ep^2 (g'(\t x))^2\right)^{-1}g'(\t x)\left(I+\epsilon f'(\t x)\right) x .
\end{eqnarray*}
By virtue of \eqref{f'g'=g'f'} and \eqref{f'f'=g'g'} we arrive at
\begin{equation}\label{aux3}
\t K(\t x,\epsilon)=   x^{\rm T}Jg'(\t x) x .
\end{equation}
Now we compute as in the previous section:
\begin{eqnarray}\label{aux4}
g'(\t x) & = & g'(x)+g''(\t x-x)=g'(x)+2\ep g''\left(I-\ep f'(x)\right)^{-1}f(x)\nonumber\\
 & = & g'(x)+2\ep g''\left(I-\ep^2 (f'(x))^2\right)^{-1}\left(I+\ep f'(x)\right)f(x).
\end{eqnarray}
We will show that the contribution to \eqref{aux3} of the terms in \eqref{aux4} with odd powers of $\epsilon$ vanishes:
\begin{eqnarray}\label{aux5}
 x^{\rm T}J g''\left(I-\ep^2 (f'(x))^2\right)^{-1}f(x)x=0.
\end{eqnarray}
For this, we use the fact that for an arbitrary vector $v\in\mathbb C^6$  we have:
\begin{equation} \label{aux7}
g''vx=g'(x)v.
\end{equation}
Indeed,  due to homogeneity of $g'(x)$,
$$
(g''vx)_i=\sum_{j,k} \frac{\partial^2 g_i}{\partial x_j\partial x_k}v_jx_k=\sum_{j}\frac{\partial g_i}{\partial x_j}v_j=(g'(x)v)_i.
$$
Due to \eqref{aux7}, the left-hand side of \eqref{aux5} is equal  to
\begin{eqnarray*}
\lefteqn{  x^{\rm T}J g'(x)\left(I-\ep^2 (f'(x))^2\right)^{-1}f(x) } \\
& = & -x^{\rm T}\nabla^2 K(x)\left(I-\ep^2 (f'(x))^2\right)^{-1}f(x)\\
& = & -2(\nabla K(x))^{\rm T}\left(I-\ep^2 (f'(x))^2\right)^{-1}f(x) \\
& = & -2(\nabla H(x))^{\rm T}\Big(B^{\rm T}\left(I-\ep^2 (f'(x))^2\right)^{-1} J\Big)\nabla H(x).
\end{eqnarray*}
One easily sees with the help of \eqref{Ham f'}, \eqref{Af'} and \eqref{cond A} that the matrix $B^{\rm T}\left(I-\ep^2 (f'(x))^2\right)^{-1} J$ is skew-symmetric, which finishes the proof of \eqref{aux5}. With this result, \eqref{aux3} turns into 
\[
\t K(\t x,\epsilon)=  x^{\rm T}Jg'(x) x +2 \epsilon^2g''\left(I-\ep^2 (f'(x))^2\right)^{-1} f'(x)f(x)x.
\]
By virtue of \eqref{f'f'=g'g'}, \eqref{f'f=g'g} and  \eqref{aux7}, we put the latter formula as
\begin{eqnarray*}
\t K(\t x,\epsilon) & = &   x^{\rm T}Jg'(x) x +2\epsilon^2 x^{\rm T}Jg''\left(I-\ep^2 (g'(x))^2\right)^{-1} g'(x)g(x)x \\
 & = & 2  x^{\rm T}Jg(x) +2\epsilon^2  x^{\rm T}Jg'(x)\left(I-\ep^2 (g'(x))^2\right)^{-1} g'(x)g(x) \\
 & = & 2 x^{\rm T} J \left(I-\ep^2 (g'(x))^2\right)^{-1} g(x) \; = \; \t K(x,\epsilon).
\end{eqnarray*}
This finishes the proof of the theorem.
\qed

We mention that, as follows from \eqref{B sum}, the three integrals $\widetilde K_i(x,\epsilon)$ corresponding to the vector fields $g_i(x)$ $(i=1,2,3)$ are related by 
$$
\widetilde K_1(x,\epsilon)+\widetilde K_2(x,\epsilon)+\widetilde K_3(x,\epsilon)=\widetilde H(x,\epsilon).
$$

%%%%%%%%%%%%%%%%%%%%%%%
%%%%%%%%%%%%%%%%%%%%%%%
\section{Invariant Poisson structure}
\label{sect symplectic}
%%%%%%%%%%%%%%%%%%%%%%%%%
%%%%%%%%%%%%%%%%%%%%%%%%%

\begin{theorem} \label{th Poisson}
The map $\Phi_f$ is Poisson with respect to the brackets with the Poisson tensor $\Pi(x)$ given by
\begin{eqnarray} \label{Pi short}
\Pi(x)  & = & J- \epsilon^2 (f'(x))^2J \label{Pi short} \\
 & = &  \big(1-\epsilon^2p_0(x)\big)J-\epsilon^2\sum_{i=1}^3 q_i(x) B_i^{\rm T}J \label{Pi long}\\
 & = & \big(1-\epsilon^2p(x)\big)J-\epsilon^2 q(x) A^{\rm T} J-\epsilon^2 r(x)(A^{\rm T})^2J, \label{Pi medium}
\end{eqnarray}
where $p_0(x)$ and $q_i(x)$ are quadratic polynomials defined in \eqref{p0}, \eqref{qi}, and 
\begin{eqnarray}
p(x) & = & p_0(x)+\sum_{i=1}^3 \alpha_iq_i(x),   \label{p}\\
q(x) & = & \sum_{i=1}^3 \beta_iq_i(x),    \label{q}\\
r(x) & = & \sum_{i=1}^3 \gamma_iq_i(x).   \label{r}
\end{eqnarray}
If the vector field $g(x)$ is associated to $f(x)$ then $\Phi_g$ is Poisson with respect to the same bracket.
\end{theorem}
\begin{proof}
First, we prove that
\beq \label{Pois prop}
d\Phi_f(x)\Pi(x) (d\Phi_f(x))^{\rm T} =\Pi(\t x).
\eeq
With the expression \eqref{Jac} for $d\Phi_f(x)$, \eqref{Pois prop} turns into
$$
\big(I+\epsilon f'(\t x)\big)\Pi(x)\big(I+\epsilon f'(\t x)\big)^{\rm T}=\big(I-\epsilon f'(x)\big)\Pi(\t x)\big(I-\epsilon f'(x)\big)^{\rm T}.
$$
Multiplying from the right by $J$ and taking into account \eqref{Ham f'}, we arrive at:
$$
\big(I+\epsilon f'(\t x)\big)\Pi(x)J\big(I-\epsilon f'(\t x)\big)=
\big(I-\epsilon f'(x)\big)\Pi(\t x)J\big(I+\epsilon f'(x) \big).
$$
According to Lemma \ref{lemma f'f'}, the matrix $\Pi(x)J$ is a linear combination of $I$, $A^{\rm T}$ and $(A^{\rm T})^2$ , therefore, by virtue of \eqref{Af'}, it commutes with $f'(\t x)$ (actually, with $f'$ evaluated at any point). Thus, the latter equation is equivalent to 
$$
\big(I-\epsilon (f'(\t x))^2\big)\Pi(x)J=\Pi(\t x)J\big(I-\epsilon (f'(x))^2\big),
$$
which is obviously true due to \eqref{Pi short} .

It remains to prove that $\Pi$ is indeed a Poisson tensor. 
\begin{lemma}
A matrix $\Pi(x)$ given by \eqref{Pi medium} is a Poisson tensor if and only if the following two conditions are satisfied:
\begin{eqnarray}
\nabla q(x) & = & C_1 \nabla r(x), \label{grad qr} \\
\nabla p(x) & = & C_2 \nabla r(x), \label{grad pr} 
\end{eqnarray}
where
\begin{eqnarray}
C_1 & = & A-aI,  \label{C1}\\
C_2 & = & A^2-a A+b I, \label{C2}
\end{eqnarray}
and
\begin{equation}
a=\tfrac{1}{2}{\rm tr}(A), \quad b=\tfrac{1}{8}({\rm tr}(A))^2-\tfrac{1}{4}{\rm tr}(A^2).
\end{equation}
\end{lemma}
\begin{proof} One has to verify the Jacobi identity
\beq \label{Jacobi id}
\{x_i,\{x_j,x_k\}\}+\{x_j,\{x_k,x_i\}\}+\{x_k,\{x_i,x_j\}\}=0
\eeq
for the 20 different triples of indices $\{i,j,k\}$ from $\{1,\ldots,6\}$. A straightforward computation based on the expression \eqref{Pi medium} for $\Pi$ (one should use a symbolic package like Maple for this computation) shows that the left-hand sides of the expressions \eqref{Jacobi id} are polynomials of order 2 in $\epsilon^2$. The system of 20 equations obtained by requiring that the coefficients by $\epsilon^2$ are equal to 0, turns out to have rank 12, with the solution given by the 12 relations $\nabla q(x)=C_1\nabla r(x)$ and $\nabla p(x)=C_2\nabla r(x)$. Then the coefficients by $\epsilon^4$ vanish identically by virtue of these 12 relations. 
\end{proof}

To finish the proof of Theorem \ref{th Poisson}, we have to prove that the functions $p(x)$, $q(x)$, $r(x)$ from \eqref{p}, \eqref{q}, \eqref{r} satisfy equations \eqref{grad qr}, \eqref{grad pr}. The following statement will be used towards this goal.

\begin{lemma} \label{lemma grads pq}
The following identities hold true:
\beq \label{grads pq}
\nabla p_0(x)=B_i \nabla q_i(x) \quad (i=1,2,3).
\eeq
\end{lemma}

From Lemma \ref{lemma grads pq} we derive:
\begin{eqnarray*}
\nabla p(x) & = &  \nabla p_0(x) +\sum_{i=1}^3 \alpha_i\nabla q_i(x) \ = \ \Big(I+\sum_{i=1}^3 \alpha_i B_i\Big) \nabla p_0(x),\\
\nabla q(x) & = & \sum_{i=1}^3 \beta_i\nabla q_i(x) \ = \ \sum_{i=1}^3 \beta_i B_i \nabla p_0(x),\\
 \nabla r(x) & = & \sum_{i=1}^3 \gamma_i\nabla q_i(x) \ = \ \sum_{i=1}^3 \gamma_i B_i \nabla p_0(x).
 \end{eqnarray*}
 So, it is sufficient to prove that
 \begin{eqnarray}
 \sum_{i=1}^3 \beta_i B_i & = & (A-aI)\sum_{i=1}^3 \gamma_i B_i , \\
 I+\sum_{i=1}^3 \alpha_i B_i & = & (A^2-aA+bI)\sum_{i=1}^3 \gamma_i B_i .
 \end{eqnarray}
 These statements are reduced to
 \begin{eqnarray}
 \sum_{i=1}^3 \beta_i B_i(\lambda) & = & (\lambda-a)\sum_{i=1}^3 \gamma_i B_i(\lambda)
   \quad {\rm for}\quad \lambda=\lambda_1,\lambda_2,\lambda_3, \label{pois lemma cond 1}\\
1+\sum_{i=1}^3 \alpha_i B_i(\lambda) & = & (\lambda^2-a\lambda+b)\sum_{i=1}^3 \gamma_i B_i(\lambda)
   \quad {\rm for}\quad \lambda=\lambda_1,\lambda_2,\lambda_3 .  \label{pois lemma cond 2}
\end{eqnarray}

Condition \eqref{pois lemma cond 1} reads:
 $$
 \beta_j+\beta_k-\beta_i=(\lambda_i-a)(\gamma_j+\gamma_k-\gamma_i),
 $$
 or, equivalently, since $\beta_1+\beta_2+\beta_3=0$ and $\gamma_1+\gamma_2+\gamma_3=0$,
 $$
 \beta_i=(\lambda_i-a)\gamma_i.
 $$
 This follows immediately from \eqref{beta}, \eqref{gamma} and the expression $a=\lambda_1+\lambda_2+\lambda_3$ from \eqref{abc}.
 
 Similarly, condition \eqref{pois lemma cond 2} reads:
 $$
1+ \alpha_j+\alpha_k-\alpha_i=(\lambda_i^2-a\lambda_i+b)(\gamma_j+\gamma_k-\gamma_i),
 $$
 or, equivalently, since $\alpha_1+\alpha_2+\alpha_3=1$,
  $$
\alpha_i-1=(\lambda_i^2-a\lambda_i+b)\gamma_i.
 $$
 Also this follows from \eqref{alpha}, \eqref{gamma} and expressions for $a$, $b$ from \eqref{abc}.
\end{proof}

\medskip

{\em Proof of Lemma \ref{lemma grads pq}.}
In this proof we write $B$ for $B_i$ and $q$ for $q_i$.
We use the characteristic property \eqref{Af'} of the matrix $f'$ which yields $B^{\rm T}f'=f' B^{\rm T}$, or, in components:
\beq\label{Af' comp}
\sum_{k}(B^{\rm T})_{jk}\frac{\partial f_k}{\partial x_m}=\sum_{k}\frac{\partial f_j}{\partial x_k}(B^{\rm T})_{km}\quad \forall j,m,
\eeq 
as well as its derivative with respect to $x_\ell$: 
\beq\label{Af'' comp}
\sum_{k}(B^{\rm T})_{jk}\frac{\partial^2 f_k}{\partial x_m\partial x_\ell}=\sum_{k}\frac{\partial^2 f_j}{\partial x_k\partial x_\ell}(B^{\rm T})_{km}\quad \forall j,m,\ell.
\eeq 

We compute the components of $\nabla q$:
\begin{equation}\label{grad q comp}
\frac{\partial q}{\partial x_\ell}=
\frac{1}{8}\frac{\partial\ {\rm tr}(B^{\rm T}(f')^2)}{\partial x_\ell}=
\frac{1}{8}\sum_{j,k,m}(B^{\rm T})_{jk}\frac{\partial^2 f_k}{\partial x_m\partial x_\ell}\frac{\partial f_m}{\partial x_j}
+\frac{1}{8}\sum_{j,k,m}(B^{\rm T})_{jk}\frac{\partial f_k}{\partial x_m}\frac{\partial^2 f_m}{\partial x_j\partial x_\ell}.
\end{equation}
The contribution to
\beq \label{A grad q comp}
(B\nabla q)_n=\sum_\ell (B^{\rm T})_{\ell n}\frac{\partial q}{\partial x_\ell}
\eeq
of the first sum in \eqref{grad q comp} is
\begin{eqnarray}
\lefteqn{\frac{1}{8}\sum_\ell (B^{\rm T})_{\ell n} \sum_{j,m}\left(\sum_k (B^{\rm T})_{jk}\frac{\partial^2 f_k}{\partial x_m\partial x_\ell}\right)\frac{\partial f_m}{\partial x_j} \qquad {\rm (use\; \eqref{Af'' comp})}}\nonumber \\
& = & \frac{1}{8}\sum_\ell (B^{\rm T})_{\ell n} \sum_{j,m}\left(\sum_k\frac{\partial^2 f_j}{\partial x_k\partial x_m} (B^{\rm T})_{k\ell}\right)\frac{\partial f_m}{\partial x_j}\nonumber\\
& = & \frac{1}{8}\sum_{j,k,m} \frac{\partial^2 f_j}{\partial x_k\partial x_m}\frac{\partial f_m}{\partial x_j}\sum_\ell(B^{\rm T})_{\ell n}(B^{\rm T})_{k\ell}
\qquad {\rm (use}\; (B^{\rm T})^2=I)
\nonumber\\
& = & \frac{1}{8}\sum_{j,k,m} \frac{\partial^2 f_j}{\partial x_k\partial x_m}\frac{\partial f_m}{\partial x_j}\delta_{kn} \; = \; 
\frac{1}{8}\sum_{j,m} \frac{\partial^2 f_j}{\partial x_n\partial x_m}\frac{\partial f_m}{\partial x_j}.
\end{eqnarray}
Similarly, the contribution of the second sum in \eqref{grad q comp} to \eqref{A grad q comp} is
\begin{eqnarray}
\lefteqn{\frac{1}{8}\sum_\ell (B^{\rm T})_{\ell n} \sum_{j,m}\left(\sum_k (B^{\rm T})_{jk}\frac{\partial f_k}{\partial x_m}\right)\frac{\partial^2 f_m}{\partial x_j\partial x_\ell} \qquad {\rm (use\; \eqref{Af' comp})}}\nonumber \\
& = & \frac{1}{8}\sum_\ell (B^{\rm T})_{\ell n} \sum_{j,m}\left(\sum_k \frac{\partial f_j}{\partial x_k}(B^{\rm T})_{km}\right)\frac{\partial^2 f_m}{\partial x_j\partial x_\ell}\nonumber\\
& = & \frac{1}{8}\sum_\ell (B^{\rm T})_{\ell n} \sum_{j,k} \frac{\partial f_j}{\partial x_k}\left(\sum_m (B^{\rm T})_{km}\frac{\partial^2 f_m}{\partial x_j\partial x_\ell}\right) \qquad {\rm (use\; \eqref{Af'' comp})}\nonumber\\
& = & \frac{1}{8}\sum_\ell (B^{\rm T})_{\ell n} \sum_{j,k} \frac{\partial f_j}{\partial x_k}\left(\sum_m \frac{\partial^2 f_k}{\partial x_j\partial x_m}
(B^{\rm T})_{m\ell}\right) \nonumber\\
& = & \frac{1}{8}\sum_{j,k,,m} \frac{\partial f_j}{\partial x_k} \frac{\partial^2 f_k}{\partial x_j\partial x_m}\sum_\ell(B^{\rm T})_{\ell n}(B^{\rm T})_{m\ell}
\qquad {\rm (use}\; (B^{\rm T})^2=I)
\nonumber\\
& = & \frac{1}{8}\sum_{j,k,m} \frac{\partial f_j}{\partial x_k} \frac{\partial^2 f_k}{\partial x_j\partial x_m}\delta_{nm} \; = \; 
\frac{1}{8}\sum_{j,k} \frac{\partial f_j}{\partial x_k} \frac{\partial^2 f_k}{\partial x_j\partial x_n} .
\end{eqnarray}
Collecting all the results, we find:
$$
(B\nabla q)_n=\frac{1}{8}\sum_{j,m} \frac{\partial^2 f_j}{\partial x_n\partial x_m}\frac{\partial f_m}{\partial x_j}+
\frac{1}{8}\sum_{j,k} \frac{\partial f_j}{\partial x_k} \frac{\partial^2 f_k}{\partial x_j\partial x_n}=\frac{1}{8}\frac{\partial\ {\rm tr}((f')^2)}{\partial x_n}=(\nabla p_0)_n,
$$
which finishes the proof. \qed
\section{Differential equations for the conserved quantities of maps $\Phi_f$, $\Phi_g$}
\label{sect diff eqs}
%%%%%%%%%%%%%%%%%%%%%%%%%%%%%%%
%%%%%%%%%%%%%%%%%%%%%%%%%%%%%%%

\begin{theorem} \label{th diff eqs}
The rational functions $\t H(x,\epsilon)$, $\t K_i(x,\epsilon)$ are related by the same differential equation as the cubic polynomials $H(x)$, $K_i(x)$:
\begin{equation}\label{grad tK}
\nabla \t K_i(x,\epsilon) =B_i \nabla \t H(x,\epsilon).
\end{equation}
As a consequence, they satisfy the same second order differential equations \eqref{d2H} as the polynomials $H(x)$, $K_i(x)$.
\end{theorem}
\begin{proof}
We start the proof with the derivation of a convenient formula for $\t H(x,\epsilon)$ and $\t K_i(x,\epsilon)$.
From \eqref{tH alt}, \eqref{Pi short}, \eqref{Pi medium} we have:
\begin{equation} \label{tH}
\t H(x,\epsilon) = 2x^{\rm T} J\left(I-\epsilon^2(f'(x))^2\right)^{-1} f(x)  = -2x^{\rm T} (\Pi(x))^{-1}f(x),
\end{equation}
and similarly
\begin{equation} \label{tK}
\t K_i(x,\epsilon) = 2x^{\rm T} J\left(I-\epsilon^2(f'(x))^2\right)^{-1} g_i(x)  = -2x^{\rm T} (\Pi(x))^{-1}g_i(x).
\end{equation}
Now we use the following result.
\begin{lemma} \label{lem Pi inv}
\begin{equation} \label{Pi inv}
(\Pi(x))^{-1}=-P_0(x,\epsilon)J-\epsilon^2\sum_{i=1}^3Q_i(x,\epsilon)B_iJ,
\end{equation}
where
\begin{eqnarray}
P_0(x,\epsilon) & = & \frac{1}{4}\left(\frac{1}{s_0(x,\epsilon)}+\frac{1}{s_1(x,\epsilon)}+\frac{1}{s_2(x,\epsilon)}+\frac{1}{s_3(x,\epsilon)}\right), \label{P0}\\
Q_1(x,\epsilon) & = & \frac{1}{4\epsilon^2}\left(-\frac{1}{s_0(x,\epsilon)}-\frac{1}{s_1(x,\epsilon)}+\frac{1}{s_2(x,\epsilon)}+\frac{1}{s_3(x,\epsilon)}\right), \label{Q1}\\
Q_2(x,\epsilon) & = & \frac{1}{4\epsilon^2}\left(-\frac{1}{s_0(x,\epsilon)}+\frac{1}{s_1(x,\epsilon)}-\frac{1}{s_2(x,\epsilon)}+\frac{1}{s_3(x,\epsilon)}\right), \label{Q2}\\
Q_3(x,\epsilon) & = & \frac{1}{4\epsilon^2}\left(-\frac{1}{s_0(x,\epsilon)}+\frac{1}{s_1(x,\epsilon)}+\frac{1}{s_2(x,\epsilon)}-\frac{1}{s_3(x,\epsilon)}\right), \label{Q3}
\end{eqnarray}
and
\begin{align}
s_0(x,\epsilon) & =  1-\epsilon^2 p_0(x)+\epsilon^2 q_1(x)+\epsilon^2 q_2(x)+\epsilon^2 q_3(x)  =1, \label{s0}\\
s_1(x,\epsilon) & =  1-\epsilon^2 p_0(x)+\epsilon^2 q_1(x)-\epsilon^2 q_2(x)-\epsilon^2 q_3(x)  =1-2\epsilon^2 q_2(x)-2\epsilon^2 q_3(x), \label{s1}\\
s_2(x,\epsilon) & =  1-\epsilon^2 p_0(x)-\epsilon^2 q_1(x)+\epsilon^2 q_2(x)-\epsilon^2 q_3(x)  =1-2\epsilon^2 q_1(x)-2\epsilon^2 q_3(x), \label{s2}\\
s_3(x,\epsilon) & =  1-\epsilon^2 p_0(x)-\epsilon^2 q_1(x)-\epsilon^2 q_2(x)+\epsilon^2 q_3(x)  =1-2\epsilon^2 q_1(x)-2\epsilon^2 q_2(x).\quad
\label{s3}
\end{align}
\end{lemma}

From \eqref{tH} and \eqref{Pi inv}, we find:
\begin{eqnarray}
\t H(x,\epsilon) & = &  2P_0(x,\epsilon)x^{\rm T} J f(x)+ 2\epsilon^2 \sum_{i=1}^3 Q_i(x,\epsilon) x^{\rm T} B_iJf(x)  \nonumber\\
& = & -2P_0(x,\epsilon)x^{\rm T} \nabla H(x) - 2\epsilon^2 \sum_{i=1}^3 Q_i(x,\epsilon) x^{\rm T} \nabla K_i(x)   \nonumber\\
& = & -6P_0(x,\epsilon)H(x)-6\epsilon^2 \sum_{i=1}^3 Q_i(x,\epsilon) K_i(x).  \label{tH thru H K}
\end{eqnarray}
By symmetry of the roles of the associated vector fields (cf. \eqref{assoc triple}), we have:
\begin{equation} \label{tK thru H K}
\t K_1(x,\epsilon) = -6P_0(x,\epsilon)K_1(x)-6\epsilon^2 Q_1(x,\epsilon)H(x) +6\epsilon^2 Q_2(x,\epsilon)K_3(x)+6\epsilon^2  Q_3(x,\epsilon)K_2(x).
\end{equation}

\begin{lemma} \label{lem grads PQ}
For the functions $P_0(x,\epsilon)$ and $Q_i(x,\epsilon)$ $(i=1,2,3)$, the following identities hold:
\beq \label{grads PQ}
\nabla P_0(x,\epsilon)=\epsilon^2 B_i\nabla Q_i(x,\epsilon).
\eeq
\end{lemma}

Now we are in a position to derive \eqref{grad tK}. For this aim, we differentiate formulas \eqref{tH thru H K}, \eqref{tK thru H K}, taking into account  differential equations $\nabla K_i=B_i\nabla H$ and \eqref{grads PQ}. We have:  
$$
-\tfrac{1}{6}\nabla \t H  =  P_0\nabla H+H\nabla P_0 +\epsilon^2\sum_{i=1}^3\big( Q_i\nabla K_i+K_i\nabla Q_i\big) 
$$
so that
\begin{eqnarray*}
-\tfrac{1}{6}B_1\nabla \t H  & = &  P_0B_1\nabla H+HB_1\nabla P_0 \\
&  &   +\epsilon^2 Q_1B_1\nabla K_1+\epsilon^2 K_1B_1\nabla Q_1\\
&  &   +\epsilon^2 Q_2B_1\nabla K_2+\epsilon^2 K_2B_1\nabla Q_2\\
&  &   +\epsilon^2 Q_3B_1\nabla K_3+\epsilon^2 K_3B_1\nabla Q_3\\
& = &  P_0\nabla K_1+\epsilon^2 H\nabla Q_1 \\
&  &   +\epsilon^2 Q_1\nabla H+K_1\nabla P_0\\
&  &   -\epsilon^2 Q_2\nabla K_3-\epsilon^2 K_2\nabla Q_3\\
&  &   -\epsilon^2 Q_3\nabla K_2-\epsilon^2 K_3\nabla Q_2\ = \ -\tfrac{1}{6}\nabla \t K_1.
\end{eqnarray*}
In this derivation we used relations like $B_1\nabla K_2=B_1B_2\nabla H=-B_3\nabla H=-\nabla K_3$ etc.
\end{proof}

{\em Proof of Lemma \ref{lem Pi inv}.} We use the ansatz \eqref{Pi inv} for $\Pi^{-1}$. Multiplying this from the left by 
$$
\Pi(x)=(1-\epsilon^2 p_0(x))J-\epsilon^2\sum_{i=1}^3 q_i(x)B_i^{\rm T}J,
$$
and using relations $B_i^{\rm T}JB_iJ=B_i^{\rm T}B_i^{\rm T}J^2=-I$ and $B_i^{\rm T}JB_jJ=B_i^{\rm T}B_j^{\rm T}J^2=B_k^{\rm T}$, we arrive at a linear combination of matrices $I$, $B_1^{\rm T}$, $B_2^{\rm T}$, $B_3^{\rm T}$. Setting the corresponding coefficients of this linear combination equal to 1, 0, 0, 0, we arrive at the linear system which can be written as
$$
\begin{pmatrix} 1-\epsilon^2 p_0 & \epsilon^2 q_1 & \epsilon^2 q_2 & \epsilon^2 q_3 \\
 \epsilon^2 q_1 & 1-\epsilon^2 p_0 & \epsilon^3 q_3 & \epsilon^2 q_2 \\
\epsilon^2 q_2 & \epsilon^2 q_3 & 1-\epsilon^2 p_0 &  \epsilon^2 q_1 \\
 \epsilon^2 q_3 & \epsilon^2 q_2 & \epsilon^2 q_1 & 1-\epsilon^2 p_0 
\end{pmatrix}
\begin{pmatrix} -P_0 \\ \epsilon^2 Q_1 \\ \epsilon^2 Q_2 \\ \epsilon^2 Q_3 \end{pmatrix}
=\begin{pmatrix} -1 \\ 0 \\ 0 \\ 0 \end{pmatrix}.
$$
One can either solve this linear system directly, or verify by a straightforward and easy check that \eqref{P0}--\eqref{Q3} with \eqref{s0}--\eqref{s3} give its unique solution. 
\qed

A remarkable by-product of this proof is the factorization of the determinant of this linear system:
\begin{equation}
\det(I-\epsilon f'(x))=s_1(x,\epsilon)s_2(x,\epsilon)s_3(x,\epsilon).
\end{equation}
\smallskip

{\em Proof of Lemma \ref{lem grads PQ}.} We note that
\begin{eqnarray*}
\nabla P_0 & = & -\frac{1}{4}\left( \dfrac{1}{s_1^2}\nabla s_1+ \dfrac{1}{s_2^2}\nabla s_2+ \dfrac{1}{s_3^2}\nabla s_3\right),\\
\nabla Q_1 & = & \frac{1}{4\epsilon^2}\left( \dfrac{1}{s_1^2}\nabla s_1- \dfrac{1}{s_2^2}\nabla s_2 -\dfrac{1}{s_3^2}\nabla s_3\right),\\
\end{eqnarray*}
and it is easy to see from \eqref{s1}, \eqref{s2}, \eqref{s3} and from \eqref{grads pq} that
$$
B_1\nabla s_1=-\nabla s_1, \quad B_1\nabla s_2=\nabla s_2, \quad B_1\nabla s_3=\nabla s_3,
$$
which yields $B_1\nabla P_0=\epsilon^2 \nabla Q_1$. \qed

%%%%%%%%%%%%%%%%%%%%%%
%%%%%%%%%%%%%%%%%%%%%%
\section{Conclusions}
%%%%%%%%%%%%%%%%%%%%%%
%%%%%%%%%%%%%%%%%%%%%%

Completely integrable Hamiltonian systems lying at the basis of our constructions, seem to be worth studying on their own. In particular, their invariant 3-dimensional varieties are intersections of three cubic hypersurfaces in the 6-dimensional space. Algebraic geometry of such varieties does not seem to be elaborated very well in the existing literature. It will be interesting to find out whether they are (affine parts) of Abelian varieties, that is, whether our systems are algebraically completely integrable. Still more interesting and intriguing are the algebraic-geometric aspects of the commuting systems of integrable maps introduced here. This will be the subject of our future research. 

\section{Acknowledgements}

This research is supported by the DFG Collaborative Research Center TRR 109 ``Discretization in Geometry and Dynamics''.

%%%%%%%%%%%%%%%%%%%%%%
%%%%%%%%%%%%%%%%%%%%%%

%%%%%%%%%%%%%%%%%%%%%%
%%%%%%%%%%%%%%%%%%%%%%

\begin{thebibliography}{}
%%%%%%%%%%%%%%%%%%%%%%
%%%%%%%%%%%%%%%%%%%%%%
\bibitem{CMOQ1}
E. Celledoni, R.I. McLachlan, B. Owren, G.R.W. Quispel.
{\em Geometric properties of Kahan's method}, 
J. Phys. A {\bf 46} (2013), 025201, 12 pp.

\bibitem{CMOQ2}
E. Celledoni, R.I. McLachlan, D.I. McLaren, B. Owren, G.R.W. Quispel.
{\em Integrability properties of Kahan's method}, 
J. Phys. A {\bf{47}} (2014), 365202, 20 pp.

\bibitem{CMOQ3}
E. Celledoni, R.I. McLachlan, D.I. McLaren, B. Owren, G.R.W. Quispel.
{\em Discretization of polynomial vector fields by polarization},
Proc. R. Soc. A {\bf 471} (2015), 20150390, 10 pp.

\bibitem{Dui}
J.J. Duistermaat. {\em Discrete Integrable Systems. QRT Maps and Elliptic Surfaces},
Springer, 2010, xii+627 pp.

\bibitem{skewHam}
H. Fa{\ss}bender, D.S. Mackey, N. Mackey, H. Xu. {\em Hamiltonian square roots of skew-Hamiltonian matrices}, 
Lin. Alg. Appl. {\bf 287} (1999), 125--159.

\bibitem{HK}
R.~Hirota, K.~Kimura.
{\em Discretization of the Euler top},
J. Phys. Soc. Japan {\bf 69} (2000), No. 3, 627--630.

\bibitem{K}
W.~Kahan.
{\em Unconventional numerical methods for trajectory calculations},
Unpublished lecture notes, 1993.

\bibitem{KH}
K.~Kimura, R.~Hirota.
{\em Discretization of the Lagrange top},
J. Phys. Soc. Japan {\bf 69} (2000), No. 10, 3193--3199.

\bibitem{PPS1}
M.~Petrera, A.~Pfadler, Yu.B.~Suris.
{\em On integrability of Hirota-Kimura-type discretizations: experimental study of the discrete Clebsch system},
Exp. Math. {\bf 18} (2009), No. 2, 223--247.

\bibitem{PPS2}
M. Petrera, A. Pfadler, Yu.B. Suris.
{\em On integrability of Hirota-Kimura type discretizations},
Regular Chaotic Dyn. {\bf 16} (2011), No. 3-4, p. 245--289. 

\bibitem{PS}
M.~Petrera, Yu.B.~Suris.
{\em On the Hamiltonian structure of Hirota-Kimura discretization of the Euler top},
 Math. Nachr. {\bf 283} (2010), No. 11, 1654--1663.

\bibitem{PS4dim} 
M. Petrera, Yu.B. Suris. 
{\em A construction of a large family of commuting pairs of integrable symplectic birational 4-dimensional maps},
{\tt arXiv:1606.08238 [nlin.SI].}
 
\bibitem{QRT}
G.R.W. Quispel, J.A.G. Roberts, C.J. Thompson. 
{\em  Integrable mappings and soliton equations II}, 
Physica D {\bf 34} (1989) 183--192.
 
 \bibitem{SS}
J.M.~Sanz-Serna. 
{\em An unconventional symplectic integrator of W. Kahan}, 
Appl. Numer. Math. {\bf 16} (1994), 245--250.


\end{thebibliography}
\end{document}